\documentclass{llncs}

\usepackage{amsmath,amssymb}
\usepackage[dvips]{epsfig}
\usepackage{epsf}
\usepackage{float}

\def\lcm{{\rm lcm}}

\def\fmod#1 #2{#1\ ({\rm mod}\ #2)}
\def\sep{{\rm sep}}
\def\nsep{{\rm nsep}}
\def \endpf{{\hfill $\Box$ \medbreak}}

\frontmatter

\author{Erik D. Demaine\inst{1}, Sarah Eisenstat\inst{1},
Jeffrey Shallit\inst{2}, and David A. Wilson\inst{1}}

\title{Remarks on Separating Words}

\institute{MIT Computer Science and Artificial Intelligence Laboratory,
32 Vassar Street,
Cambridge, MA 02139,
USA,
\email{\{edemaine,seisenst,dwilson\}@mit.edu}
\and
University of Waterloo,
Waterloo, ON  N2L 3G1, Canada,
\email{shallit@cs.uwaterloo.ca}
}

\begin{document}

\newtheorem{openproblem}{Open Problem}

\maketitle

\begin{abstract}
The {\it separating words problem} asks for the size of the smallest
DFA needed to distinguish between two words of length $\leq n$ (by
accepting one and rejecting the other).  In this paper we survey what
is known and unknown about the problem, consider some
variations, and prove several new results.
\end{abstract}

\section{Introduction}

Imagine a computing device with very limited powers.
What is the simplest computational problem you could ask it to solve?
It is not the addition of two numbers, nor sorting, nor string
matching --- it is telling two inputs apart:  distinguishing them in
some way.

Take as our computational model the deterministic finite automaton or
DFA.  As usual, it consists of a $5$-tuple, $M = (Q, \Sigma, \delta,
q_0, F)$, where $Q$ is a finite nonempty set of states,
$\Sigma$ is a nonempty input alphabet, $\delta: Q \times \Sigma \rightarrow
Q$ is the transition function (assumed to be {\it complete}, or defined
on all members of its domain), $q_0 \in Q$ is the initial state, and $F 
\subseteq Q$ is a set of final states.  

We say that a DFA $M$ {\it separates} $w$ and $x$ if $M$ accepts one but rejects
the other.
Given two distinct words $w, x$ we let $\sep(w,x)$ be the number of states
in the smallest DFA accepting $w$ and rejecting $x$.
For example, the DFA below separates $0010$ from $1000$.
\begin{figure}[H]
\begin{center}
\epsfig{file=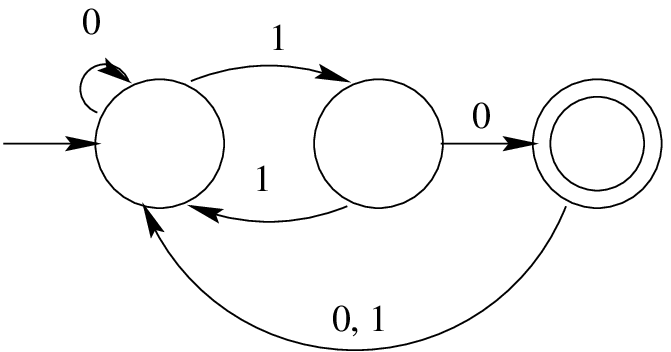,height=4cm}
\end{center}
\label{data140}
\end{figure}
However, by a brief computation, we see that
no $2$-state DFA can separate these two words.  So
$\sep(1000,0010) = 3$.
Note that $\sep(w, x) = \sep(x, w)$,
because the language of a DFA can be complemented
by swapping the reject and accept states.

We let $S(n) = \displaystyle\max_{{w\not= x} \atop {|w|, |x| \leq n}} \sep(w,x)$.
The {\it separating words problem} is to determine
good upper and lower bounds on $S(n)$.
This problem
was introduced 25 years ago by Goral\v c\'{\i}k
and Koubek \cite{Goralcik&Koubek:1986},
who proved $S(n) = o(n)$.
It was later studied by Robson \cite{Robson:1989,Robson:1996}, who
obtained the best upper bound so far:  $S(n) = O(n^{2/5} (\log n)^{3/5})$.

As an additional motivation, the separating words problem can be
viewed as an inverse of a classical problem from the early days of
automata theory: given two DFAs accepting different languages, what
length of word suffices to distinguish them?  More precisely, given
two DFAs $M_1$ and $M_2$, with $m$ and $n$ states, respectively, with
$L(M_1) \not= L(M_2)$, what is a good bound on the length of the
shortest word accepted by one but not the other?  The usual
cross-product construction quickly gives an upper bound of $mn-1$
(make a DFA for $L(M_1) \ \cap \ \overline{L(M_2)}$).  But the optimal
upper bound of $m+n-2$ follows from the usual algorithm for minimizing
automata.  Furthermore, this bound is best possible \cite[Thm.\
3.10.6]{Shallit:2009}.  For NFAs the bound is exponential in $m$ and
$n$ \cite{Nozaki:1979}.

From the following result, already proved by Goral\v c\'{\i}k
and Koubek \cite{Goralcik&Koubek:1986},
we know that the hard case of word separation
comes from words of equal length:

\begin{proposition}
Suppose $|w|, |x| \leq n$ and $|w| \not= |x|$.  Then 
$\sep(w,x) = O(\log n)$. Furthermore, there is an infinite class of
examples where $\sep(w,x) = \Omega(\log n)$.
\label{one}
\end{proposition}

We use the following lemma \cite{Shallit&Breitbart:1996}:

\begin{lemma}
If $0 \leq i, j \leq n$ and $i \not= j$, then there is a prime
$p \leq 4.4 \log n$ such that $i \not\equiv \fmod{j} {p}$.
\label{pnt}
\end{lemma}

\begin{proof} (of Proposition \ref{one})
If $|w| \not= |x|$, then
by Lemma~\ref{pnt}
there exists a prime $p \leq 4.4 \log n$ such that ${|w| \bmod p}
\not= {|x| \bmod p}$.  
Hence a simple cycle of $p$ states serves to
distinguish $w$ from $x$.

On the other hand, no DFA with $n$
states can distinguish
$$ 0^{n-1} \ \ \ \text{from} \ \ \ 0^{n-1+\lcm(1,2,\ldots, n)}. $$
To see this, 
let $p_i = \delta(q_0, 0^i)$ for $i \geq 0$.  Then $p_i$ is
ultimately periodic with period $\leq n$ and preperiod at most $n-1$.
Thus $p_{n-1} = p_{n-1+\lcm(1,2,\ldots, n)}$.
Since $\lcm(1,2, \ldots, n) = e^{n(1+o(1))}$
by the prime number theorem, the $\Omega(\log n)$ lower bound follows.
\endpf
\end{proof}

As an example, 
suppose $|w| = 22$ and $|x| = 52$.  Then $|w| \equiv \fmod{1} {7}$ and
$|x| \equiv \fmod{3} {7}$.  So we can accept $w$ and reject $x$ with a DFA that uses a cycle of
size $7$, as follows:

\begin{figure}[H]
\begin{center}
\input cycle1.pstex_t
\end{center}
\label{cycle1}
\end{figure}

In what follows, then, we only consider the case of equal-length words,
and we redefine $S(n) = \displaystyle\max_{{w\not= x}\atop {|w|=|x|=n}} \sep(w,x)$.
The goal of the paper is to survey what is known and unknown, and
to examine some variations on the original problem.  Our main new
results are Theorems~\ref{hamming} and \ref{nondets}.

\section{Independence of alphabet size}

As we have defined it, $S(n)$ could conceivably depend on the size of
the alphabet $\Sigma$.
Let $S_k (n)$ be the maximum number of states needed to separate
two length-$n$ words over an alphabet of size $k$.
Then we might have a different value $S_k(n)$ depending on $k = |\Sigma|$.
The following result shows this is not the case for $k \geq 2$.
This result was stated in \cite{Goralcik&Koubek:1986} without proof;
we supply a proof here.

\begin{proposition}
For all $k \geq 2$ we have $S_k(n) = S_2(n)$.
\end{proposition}

\begin{proof}
Suppose $x, y$ are distinct length-$n$ words 
over an alphabet $\Sigma$ of size $k > 2$. 
Then $x$ and $y$ must differ in some position, say 
for $a \not= b$,
\begin{eqnarray*}
x & = & x' \ a \ x'',\\
y & = & y' \ b \ y'',
\end{eqnarray*}
for $|x'| = |y'|$.

Now map $a$ to $0$, $b$ to $1$ and map all other letters of $\Sigma$
to $0$.
This gives two new distinct binary words $X$ and $Y$ of length $n$.
If $X$ and $Y$ can be separated by an $m$-state DFA, then
so can $x$ and $y$, by renaming transitions of the DFA
to be over $\Sigma \backslash b$ and $\{b\}$ instead of $0$ and $1$, respectively.
Thus $S_k (n) \leq S_2 (n)$.  But clearly $S_2 (n) \leq S_k (n)$, since
every binary word can be considered as a word over the larger alphabet
$\Sigma$.  So $S_k (n) = S_2 (n)$.
\endpf
\end{proof}

\section{Average case}

One frustrating aspect of the separating words problem is that nearly
all pairs of words can be easily separated.    This means that bad
examples cannot be easily produced by random search.

\begin{proposition}
  Consider a pair of words $(w,x)$ selected uniformly from the set of
  all pairs of unequal words of length $n$ over an alphabet of size
  $k$.  Then the expected number of states needed to separate $w$ from
  $x$ is $O(1)$.
\end{proposition}

\begin{proof}
With probability $1-1/k$, two randomly-chosen words will differ in the
first position, which can be detected by an automaton with $3$ states.
With probability $(1/k)(1-1/k)$ the words will agree in the
first position, but differ in the second, etc.  Hence the expected
number of states needed to distinguish two randomly-chosen words
is bounded by $\sum_{i \geq 1} (i+2)(1/k)^{i-1} (1-1/k) = 
(3k-2)/(k-1) \leq 4$.
\endpf
\end{proof}

\section{Lower bounds for words of equal length}

First of all, there is a lower bound analogous to that in
Proposition~\ref{one} for words of {\it equal} length.  This does not
appear to have been known previously.

\begin{theorem}
No DFA of at most $n$ states can separate the equal-length
binary words $w = 0^{n-1} 1^{n-1+\lcm(1,2,\ldots, n)}$ and
$x = 0^{n-1+\lcm(1,2,\ldots,n)} 1^{n-1}$.
\label{equal}
\end{theorem}

\begin{proof}
In pictures, we have
\begin{figure}[H]
\begin{center}
\epsfig{file=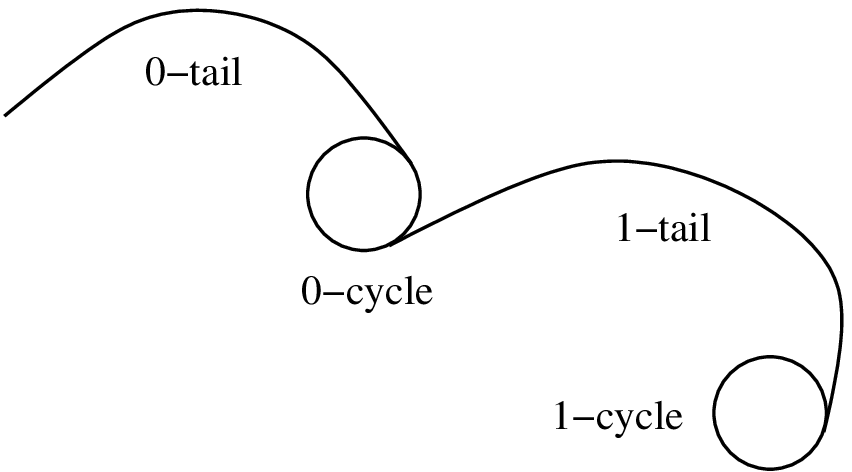,height=3cm}
\end{center}
\label{data135}
\end{figure}

More formally, 
let $M$ be any DFA with $n$ states, let $q$ be any state,
and let $a$ be any letter.  
Let $p_i = \delta(q, a^i)$ for $i \geq 0$.  Then $p_i$ is
ultimately periodic with period $\leq n$ and preperiod (``tail'') at most $n-1$.
Thus $p_{n-1} = p_{n-1+\lcm(1,2,\ldots, n)}$.

It follows that after processing $0^{n-1}$ and
$0^{n-1+\lcm(1,2,\ldots,n)}$, $M$ must be in the same state.  Similarly,
after processing
$0^{n-1} 1^{n-1+\lcm(1,2,\ldots, n)}$ and
$0^{n-1+\lcm(1,2,\ldots,n)} 1^{n-1}$, $M$ must be in the same state.
So no $n$-state machine can separate $w$ from $x$.
\endpf
\end{proof}

We now prove a series of very simple results showing that if $w$ and
$x$ differ in some ``easy-to-detect'' way, then $\sep(w,x)$ is small.

\subsection{Differences Near the Beginning or End of Words}

\begin{proposition}
Suppose $w$ and $x$ are words that
differ in some symbol that occurs
$d$ positions from the start.
Then $\sep(w,x) \leq d+2$.
\end{proposition}

\begin{proof}
Let $t$ be a prefix of length $d$ of $w$.  Then $t$ is not a prefix of
$x$.  We can accept the language $t\Sigma^*$ using $d+2$ states; such
an automaton accepts $w$ and rejects $x$.
\endpf
\end{proof}

For example, to separate 
$$01010011101100110000$$
from
$$01001111101011100101$$
we can build a DFA to recognize words that begin with $0101$:

\begin{figure}[H] 
\begin{center} 
\input x0101.pstex_t
\end{center}
\label{x0101}
\end{figure}
\noindent (Transitions to a dead state are omitted.)

\begin{proposition}
Suppose $w$ and $x$ differ in some symbol that occurs $d$ positions from the
end.  Then $\sep(w,x) \leq d+1$.
\end{proposition}

\begin{proof}
Let the DFA $M$ be the usual
pattern-recognizing automaton for the length-$d$ suffix $s$ of $w$,
ending in an accepting state if the suffix is recognized.
Then $M$ accepts $w$ but rejects $x$.    States of $M$ correspond to
prefixes of $s$, and $\delta(t, a) = $ the longest suffix of $ta$ that
is a prefix of $s$.
\endpf
\end{proof}

For example, to separate 
$$ 11111010011001010101$$
from
$$ 11111011010010101101$$
we can build a DFA to recognize those words that end in $0101$:

\begin{figure}[H] 
\begin{center} 
\input suff0101.pstex_t
\end{center}
\label{suff0101}
\end{figure}

\subsection{Fingerprints}

Define $|w|_a$ as the number of occurrences of the symbol $a$ in the word
$w$.

\begin{proposition}
If $|w|, |x| \leq n$
and $|w|_a \not= |x|_a$ for some symbol $a$, then $\sep(w,x) = O(\log n)$.
\end{proposition}

\begin{proof}
By the prime number theorem, if $|w|, |x| = n$, and $w$ and $x$ have
$k$ and $m$ occurrences of $a$ respectively ($k \not= m$),
then there is a prime $p = O(\log n)$ such that $k \not\equiv m$ (mod $p$).
So we can separate $w$ from $x$ just by counting the number of $a$'s,
modulo $p$.
\endpf
\end{proof}

Analogously, we have the following result.

\begin{proposition}
If there is a pattern of length $d$ occurring a differing number of
times in $w$ and $x$, with $|w|, |x| \leq n$, then 
$\sep(w,x) = O(d \log n)$.
\end{proposition}

\subsection{Pairs with Low Hamming Distance}

The previous results have shown that if $w$ and $x$ have differing
``fingerprints'', then they are easy to separate.
By contrast, the next result shows that if $w$ and $x$ are very
similar, then they are also easy to separate.

The \textit{Hamming distance} $H(w,x)$ between two equal-length
words $w$ and $x$ is defined to be the number of
positions where they differ.  

\begin{theorem}
Let $w$ and $x$ be words of length $n$.
If $H(w,x) \leq d$, then $\sep(w,x) = O(d \log n)$.
\label{hamming}
\end{theorem}

\begin{proof}
Without loss of generality, assume $x$ and $y$ are binary words, and
$x$ has a $1$ in some position where $y$ has a $0$.
Consider the following picture:

\begin{figure}[H] 
\begin{center} 
\input sep1.pstex_t
\end{center}
\label{double1}
\end{figure}

Let $i_1 < i_2 < \ldots < i_d$ be the positions where $x$ and $y$ differ.
Now consider $N = (i_2 - i_1)(i_3 - i_1) \cdots (i_d - i_1)$.
Then $N < n^{d-1}$.
By the prime number theorem, there exists some prime $p = O(\log N) = 
O(d \log n)$ such that $N$ is not divisible by $p$.
So $i_j \not\equiv \fmod{i_1} {p}$ for $2 \leq j \leq d$.  

Define $a_{p,k}(x)$ = $\left( \displaystyle\sum_{j \equiv \fmod{k} {p}}
  x_j \right) \bmod 2$.  This value can be calculated by a DFA
consisting of two connected rings of $p$ states each.
We use such a DFA calculating $a_{p,i_1}$.  Since $p$ is not a factor
of $N$, none of the positions $i_2, i_3, \ldots, i_d$ are included in
the count $a_{p,i_1}$, and the two words $x$ and $y$ agree in all
other positions.  So $x$ contains exactly one more $1$ in these
positions than $y$ does, and hence we can separate the two words using
$O(d \log n)$ states.
\endpf
\end{proof}

\section{Special classes of words}

\subsection{Reversals}

It is natural to think that pairs of words that are related might be
easier to separate than arbitrary words; for example, it might be easy
to separate a word from its reversal.  No better upper bound is known
for this special case.  However, we still have a lower bound of
$\Omega(\log n)$ for this restricted problem:

\begin{proposition}
There exists a class of words $w$ for which $\sep(w, w^R) = 
\Omega(\log n)$ where $n = |w|$.
\end{proposition}

\begin{proof}
Consider separating 
$$w = 0^{t-1} 1 0^{t-1+\lcm(1,2, \ldots t)}$$
from
$$w^R = 0^{t-1+\lcm(1,2, \ldots t)} 1 0^{t-1}.$$
Then, as before,
no DFA with $\leq t$ states can separate $w$ from $w^R$.
\endpf
\end{proof}

Must $\sep(w^R, x^R) = \sep(w,x)$?
No, for $w = 1000$, $x = 0010$, we have
$$\sep(w,x) = 3$$
but $$\sep(w^R, x^R) = 2.$$

\begin{openproblem}
Is $ \left| \sep(x,w) - \sep(x^R,w^R) \right|$ unbounded?
\end{openproblem}

\subsection{Conjugates}

Two words $w, w'$ are \textit{conjugates} if one is a cyclic shift of the
other.
For example, the English words \texttt{enlist} and \texttt{listen} are
conjugates.
Is the separating words problem any easier if restricted to
pairs of conjugates?

\begin{proposition}
There exist a infinite class of pairs of words $w, x$ such that
$w, x$ are conjugates, and $\sep(w,x)= \Omega(\log n)$ for
$|w|=|x|=n$.
\end{proposition}

\begin{proof}
Consider again
$$w = 0^{t-1} 1 0^{t-1+\lcm(1,2, \ldots t)} 1$$
and
$$w' = 0^{t-1+\lcm(1,2, \ldots t)} 1 0^{t-1} 1.$$
\endpf
\end{proof}

\section{Nondeterministic separation}

We can define $\nsep(w,x)$ in analogy with $\sep$:  the number of states
in the smallest NFA accepting $w$ but rejecting $x$.  There do not seem
to be any published results about this measure.

Now there is an asymmetry in the inputs:  $\nsep(w,x)$ need not
equal $\nsep(x,w)$.
For example, the following $2$-state NFA accepts $w = 000100$ and
rejects $x = 010000$, so $\nsep(w,x) \leq 2$.

\begin{figure}[H]
\begin{center}
\input nsep.pstex_t
\end{center}
\label{asymm}
\end{figure}

However, an easy computation shows that
there is no $2$-state NFA accepting $x$ and rejecting $w$,
so $\nsep(x,w) \geq 3$.

\begin{openproblem}
Is $|\nsep(x,w) - \nsep(w,x)|$ unbounded?
\end{openproblem}

A natural question is whether NFAs give more separation power than
DFAs.  Indeed they do, since $\sep(0001,0111) = 3$ but
$\nsep(0001,0111) = 2.$  However, a more interesting question is the
\emph{extent} to which nondeterminism helps with separation --- for
example, whether it contributes only a constant factor or there is any
asymptotic improvement in the number of states required.

\begin{theorem}
The quantity $\sep(w,x)/\nsep(w,x)$ is unbounded.
\label{nondets}
\end{theorem}

\begin{proof}
Consider once again the words
$$ w =  0^{t-1+\lcm(1,2,\ldots,t)} 1^{t-1} \ \ \ {\rm and} \ \ \ 
 x = 0^{t-1} 1^{t-1+\lcm(1,2,\ldots,t)}  $$
where $t = n^2 - 3n + 2$, $n \geq 4$.  

We know from Theorem~\ref{equal}
that any DFA separating these words must have at
least $t+1 = n^2 - 3n+3$ states.

Now consider the following NFA $M$:

\begin{figure}[H]
\begin{center}
\input aut11.pstex_t
\end{center}
\label{data136}
\end{figure}

The language accepted by this NFA is $\lbrace 0^a \ : \ a \in A \rbrace 1^*$,
where $A$ is the set of all integers representable by a non-negative integer
linear combination of $n$ and $n-1$.
But  $t-1 = n^2 - 3n + 1 \not\in A$, as can be
seen by computing $t-1$ modulo $n-1$ and modulo $n$.
On the other hand, every integer $\geq t$ is in $A$.
Hence $w = 0^{t-1+\lcm(1,2,\ldots,t)} 1^{t-1}  $ is accepted by $M$ but 
$x = 0^{t-1} 1^{t-1+\lcm(1,2,\ldots,t)}$ is not.

Now $M$ has $2n = \Theta(\sqrt{t})$ states, so  $\sep(x,w) / \nsep(x,w) \geq
\sqrt{t} = \Omega(\sqrt{\log |x|})$, which is unbounded.
\endpf
\end{proof}

\begin{openproblem}
Find better bounds on $\sep(w,x)/\nsep(w,x)$.
\end{openproblem}

We can also get an $\Omega(\log n)$ lower bound for nondeterministic
separation.

\begin{theorem}
No NFA of $n$ states can separate
$$0^{n^2-1} 1^{n^2-1+\lcm(1,2,\ldots, n)}$$
from
$$ 0^{n^2-1+\lcm(1,2,\ldots, n)} 1^{n^2-1} .$$
\end{theorem}

\begin{proof}
A result of Chrobak \cite{Chrobak:1986}, as corrected by
To \cite{To:2009}, states that every unary $n$-state
NFA is equivalent to one
consisting of a ``tail'' of at most $O(n^2)$ states, followed by a single
nondeterministic state that leads to a set of cycles, each of which has
at most $n$ states.  The size of the tail was proved to be at most
$n^2-2$ by Geffert \cite{Geffert:2007}.

Now we use the same argument as for DFAs above.
\endpf
\end{proof}

\begin{openproblem}
Find better bounds on $\nsep(w,x)$ for $|w|=|x| = n$, as a function of $n$.
\end{openproblem}

\begin{theorem}
We have $\nsep(w,x) =
\nsep(w^R, x^R)$.
\end{theorem}

\begin{proof}
  Let $M$ be an NFA with the smallest number of states accepting $w$
  and rejecting $x$.  Now make a new NFA $M'$ with initial state equal
  to any one element of $\delta(q_0, w)$ and final state $q_0$, and
  all other transitions of $M$ reversed.  Then $M'$ accepts $w^R$.
  But $M'$ rejects $x^R$.  For if $M'$ accepted $x^R$ then $M$ would
  also accept $x$, since the input string and transitions are
  reversed.
\endpf
\end{proof}

\section{Separation by 2DPDA's}

In \cite{Currie&Petersen&Robson&Shallit:1999}, the authors showed that
words can be separated with small context-free
grammars (and hence small PDA's).  In this section we observe 

\begin{proposition}
Two distinct
words of length $n$ can be separated by a 2DPDA of size $O(\log n)$.
\end{proposition}

\begin{proof}
Recall that
a 2DPDA is a deterministic pushdown automaton, with endmarkers surrounding
the input, and two-way access to the input tape.  Given distinct
strings $w, x$ of length $n$, they must differ in some position $p$
with $1 \leq p \leq n$.  Using $O(\log p)$ states, we can reach position
$p$ on the input tape and accept if (say) the corresponding character
equals $w[p]$, and reject otherwise.

Here is how to access position $p$ of the input.  We show how to go
from scanning position $i$ to position $2i$ using a constant number
of states:  we move left on the input, pushing two symbols per move 
on the stack, until the left endmarker is reached.  Now we move right,
popping one symbol per move, until the initial stack symbol is reached.
Using this as a subroutine, and applying it to the binary expansion of
$p$, we can, using $O(\log p)$ states, reach
position $p$ of the input.
\endpf
\end{proof}

\section{Permutation automata}

We conclude by relating the separating words problem to a natural problem
of algebra.

Instead of arbitrary automata, we could restrict our attention to
automata where each letter induces a permutation of the states (``permutation
automata''), as suggested by Robson \cite{Robson:1996}.  He obtained an
$O(n^{1/2})$ upper bound in this case.

For an $n$-state automaton, the action of each letter can be viewed as
an element of $S_n$, the symmetric group on $n$ elements.

Turning the problem around, then, we could ask:  what is the shortest pair
of distinct equal-length binary words $w, x$, such that for all morphisms
$\sigma:\lbrace 0,1 \rbrace^* \rightarrow S_n$ we have
$\sigma(w) = \sigma(x)$?
Although one might suspect that the answer is $\lcm(1,2,\ldots, n)$,
for $n = 4$, there is a shorter pair (of length 11):
$00000011011$ and $11011000000$.

Now if $\sigma(w) = \sigma(x)$ for all $\sigma$, then  (if we define
$\sigma(x^{-1}) = \sigma(x)^{-1}$) we have that
$\sigma(wx^{-1}) = $ the identity permutation for all $\sigma$.

Call any nonempty word $y$ over the letters $0, 1, 0^{-1}, 1^{-1}$ an 
\textit{identical relation} if $\sigma(y) = $ the identity for all
morphisms $\sigma$.  
We say $y$ is \textit{nontrivial} if $y$ contains no occurrences of
$0 0^{-1}$ and $1 1^{-1}$.

What is the length $\ell$ of the shortest nontrivial identical relation over
$S_n$?
Recently Gimadeev and Vyalyi \cite{Gimadeev&Vyalyi:2010}
proved $\ell = 2^{O(\sqrt{n} \log n)}$.

\section{Acknowledgments}

Thanks to Martin Demaine and Anna Lubiw, who participated in the
open problem session at MIT where these problems were discussed.

\newcommand{\noopsort}[1]{} \newcommand{\singleletter}[1]{#1}

\end{document}